\documentclass[conference]{IEEEtran}
\usepackage{multirow}
\usepackage{epsfig}
\usepackage{subfigure}
\usepackage{cite}
\usepackage{amsmath}
\usepackage{amssymb}
\usepackage{footnote}
\usepackage{array}
\usepackage{subfigure, verbatim, animate}
\usepackage{comment, amsmath, amsfonts, amsthm}
\usepackage[mathscr]{euscript}
\usepackage{epsfig}
\usepackage{amsfonts,amssymb,latexsym,amsmath, amsxtra,verbatim}
\usepackage[all]{xy}
\usepackage{auto-pst-pdf}%%%
\usepackage{psfrag}
\usepackage{xcolor}
\usepackage{xspace}
\usepackage{bbm}
\catcode`@=11
\chardef\mathlig@atcode\count255

\def\actively#1#2{\begingroup\uccode`\~=`#2\relax\uppercase{\endgroup#1~}}
\def\mathlig@gobble{\afterassignment\mathlig@next@cmd\let\mathlig@next= }
\def\mathlig@delim{\mathlig@delim}
\def\mathlig@defcs#1{\expandafter\def\csname#1\endcsname}
\def\mathlig@let@cs#1#2{\expandafter\let\expandafter#1\csname#2\endcsname}
\def\mathlig@appendcs#1#2{\expandafter\edef\csname#1\endcsname{\csname#1\endcsname#2}}

\def\mathlig#1#2{\mathlig@checklig#1\mathlig@end\mathlig@defcs{mathlig@back@#1}{#2}\ignorespaces}

\def\mathlig@checklig#1#2\mathlig@end{%
 \expandafter\ifx\csname mathlig@forw@#1\endcsname\relax
 \expandafter\mathchardef\csname mathlig@back@#1\endcsname=\mathcode`#1%
 \mathcode`#1"8000\actively\def#1{\csname mathlig@look@#1\endcsname}%
 \mathlig@dolig#1\mathlig@delim
\fi
\mathlig@checksuffix#1#2\mathlig@end
}

\def\mathlig@checksuffix#1#2\mathlig@end{%
\ifx\mathlig@delim#2\mathlig@delim\relax\else\mathlig@checksuffix@{#1}#2\mathlig@end\fi
}
\def\mathlig@checksuffix@#1#2#3\mathlig@end{%
\expandafter\ifx\csname mathlig@forw@#1#2\endcsname\relax\mathlig@dosuffix{#1}{#2}\fi
\mathlig@checksuffix{#1#2}#3\mathlig@end
}

\def\mathlig@dosuffix#1#2{%
\mathlig@appendcs{mathlig@toks@#1}{#2}%
\mathlig@dolig{#1}{#2}\mathlig@delim
}

\def\mathlig@dolig#1#2\mathlig@delim{%
 \mathlig@defcs{mathlig@look@#1#2}{%
 \mathlig@let@cs\mathlig@next{mathlig@forw@#1#2}\futurelet\mathlig@next@tok\mathlig@next}%
 \mathlig@defcs{mathlig@forw@#1#2}{%
  \mathlig@let@cs\mathlig@next{mathlig@back@#1#2}%
  \mathlig@let@cs\checker{mathlig@chck@#1#2}%
  \mathlig@let@cs\mathligtoks{mathlig@toks@#1#2}%
  \expandafter\ifx\expandafter\mathlig@delim\mathligtoks\mathlig@delim\relax\else
  \expandafter\checker\mathligtoks\mathlig@delim\fi
  \mathlig@next
 }%
 \mathlig@defcs{mathlig@toks@#1#2}{}%
 \mathlig@defcs{mathlig@chck@#1#2}##1##2\mathlig@delim{%
  \ifx\mathlig@next@tok##1%
   \mathlig@let@cs\mathlig@next@cmd{mathlig@look@#1#2##1}\let\mathlig@next\mathlig@gobble
  \fi 
  \ifx\mathlig@delim##2\mathlig@delim\relax\else
   \csname mathlig@chck@#1#2\endcsname##2\mathlig@delim
  \fi
 }%
 \ifx\mathlig@delim#2\mathlig@delim\else
  \mathlig@defcs{mathlig@back@#1#2}{\csname mathlig@back@#1\endcsname #2}%
 \fi
}%

\catcode`@\mathlig@atcode

\newcommand{\muspace}{\mspace{1mu}}

\DeclareRobustCommand{\scond}{\mathchoice{\muspace\vert\muspace}{\vert}{\vert}{\vert}}
\mathlig{|}{\scond}

\DeclareRobustCommand{\discint}{\mathchoice{\mspace{-1.5mu}:\mspace{-1.5mu}}{\mspace{-1.5mu}:\mspace{-1.5mu}}{:}{:}}
\mathlig{::}{\discint}

\def\var{\mathop{\rm Var}\nolimits}%
\def\co{\mathop{\rm co}\nolimits}%

\newcommand{\Cc}{\mathcal{C}}

\newcommand{\Pc}{\mathcal{P}}

\newcommand{\Sc}{\mathcal{S}}

\newcommand{\Xc}{\mathcal{X}}
\newcommand{\Yc}{\mathcal{Y}}

\newcommand{\Cr}{\mathscr{C}}

\newcommand{\Rr}{\mathscr{R}}

\def\a{\alpha}

\newcommand{\U}{\mathrm{Unif}}

\def\textiid{i.i.d.\@\xspace}
\newcommand\iid{\ifmmode\text{ i.i.d. } \else \textiid \fi}

\def\mathllap{\mathpalette\mathllapinternal}
\def\mathllapinternal#1#2{%
  \llap{$\mathsurround=0pt#1{#2}$}}

\def\clap#1{\hbox to 0pt{\hss#1\hss}}
\def\mathclap{\mathpalette\mathclapinternal}
\def\mathclapinternal#1#2{%
  \clap{$\mathsurround=0pt#1{#2}$}}

\let\oldstackrel\stackrel
\renewcommand{\stackrel}[2]{\oldstackrel{\mathclap{#1}}{#2}}

\renewcommand{\hbar}{h\mathllap{\overline{\vphantom{h}\hphantom{\rule{4.6pt}{0pt}}}\mspace{0.77mu}}}

\catcode`~=11 % make LaTeX treat tilde (~) like a normal character
\newcommand{\urltilde}{\kern -.06em\lower -.06em\hbox{~}\kern .02em}
\catcode`~=13 % revert back to treating tilde (~) as an active character

\newtheorem{theorem}{\textbf{Theorem}}

\newtheorem{lemma}{\textbf{Lemma}}

\newtheorem{proposition}{\textbf{Proposition}}

\newtheorem{example}{\textbf{Example}}

\newcommand{\bq}{\bar{p}_2}
\newcommand{\bp}{\bar{p}_1}
\newcommand{\q}{p_2}
\newcommand{\p}{p_1}
\newcommand{\argmax}{\operatornamewithlimits{arg\,max}}

\IEEEoverridecommandlockouts
\makeatother
\begin{document}
\title{Capacity Region of the Broadcast Channel with Two Deterministic Channel State Components}

\author{Hyeji Kim and Abbas El Gamal\\
Department of Electrical Engineering\\
 Stanford University\\
 Email: hyejikim@stanford.edu, abbas@ee.stanford.edu
\thanks{ This work was partially supported by Air Force grant FA9550-10-1-0124.   }%
}

\maketitle
\begin{abstract} This paper establishes the capacity region of a class of broadcast channels with random state in which each channel component is selected from two possible functions and each receiver knows its state sequence. This channel model does not fit into any class of broadcast channels for which the capacity region was previously known and is useful in studying wireless communication channels when the fading state is known only at the receivers. The capacity region is shown to coincide with the UV outer bound and is achieved via Marton coding.
\end{abstract}

\section{Introduction}
The 2-receiver {\em  broadcast channel with two deterministic channel states} (or BC-TDCS in short) is a discrete memoryless broadcast channel with random state $(\Xc \times \Sc, p(s)p(y_1,y_2|x,s), \Yc_1 \times \Yc_2)$, where $S=(S_1,S_2)\in \{1,2\}^2$,  $p_{S_1}(1)=p_1, p_{S_1}(2)=1-p_1=\bar p_1$ and $p_{S_2}(1)=p_2, p_{S_2}(2)=\bar p_2$, and the outputs %Jan14: Moved the position of $(X by S, ~, ~)$
\begin{align*}
Y_1 = \begin{cases}	f_1(X) &\text{ if } S_1=1,\\ 
				f_2(X) &\text{ if } S_1=2,
	\end{cases}\\
Y_2 = \begin{cases} f_1(X) &\text{ if } S_2=1,\\
				f_2(X) &\text{ if } S_2= 2
	\end{cases}
\end{align*}
for some deterministic functions $f_1$ and $f_2$ of the input $X$. As an example of a BC-TDCS, consider the following. %Jan14: Separated notations and definition.
%Jan15: CT's comment: Perhaps you can mention that you use a shorthand f_1 for f_1(X), f_2 for f_2(X). It might be a bit confusing to see a function used as a random variable, since a function is not a random object. So adding a line explaining this shorthand may be helpful.

\begin{example}[Blackwell channel with state~\cite{dstate}]\label{ex:blackwell}
\textnormal{The functions $f_1$ and $f_2$ for this example are depicted in Figure~\ref{fig:bc-blackwell}.
\begin{figure}[h]
\begin{center}
\psfrag{0}[r]{0}
\psfrag{1}[r]{2}
\psfrag{2}[r]{1}
\psfrag{X}[c]{$X$}
\psfrag{Y1}[c]{$f_1(X)$}
\psfrag{Y2}[c]{$f_2(X)$}
\psfrag{z}[l]{0}
\psfrag{o}[l]{1}
\includegraphics[scale=0.32]{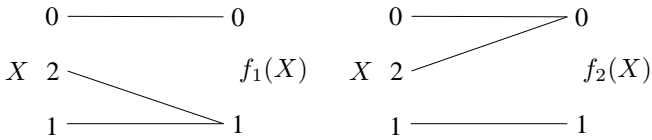}
\caption{The deterministic components of the Blackwell channel with state.} 
%Jan14: Orginally, it was: \caption{The deterministic channel states of the Blackwell channel with state.}
\label{fig:bc-blackwell}
\end{center}
\end{figure}}
 \end{example}
\vspace*{-.25in} 
In this paper, we consider the setup in which the sender wishes to transmit an independent message $M_j \in [1:2^{nR_j}]$ to receiver $j\in \{1,2\}$ and receiver $j$ knows the state sequence $S_j^n$ but the sender does not. %Jan14:Before: the state sequence $S_j^n$ is known at receiver $j$ but the sender has no knowledge of the state. 
We define achievable rate pairs $(R_1,R_2)$ in the standard way~\cite{El-Gamal--Kim2011} and the capacity region $\Cr$ as the closure of the set of all achievable rate pairs.
%Jan14:Before: the state sequence $S_j^n$ is known at receiver $j$ but the sender has no knowledge of the state. We define the capacity region $\Cr$ as the closure of the set of achievable rate pairs $(R_1,R_2)$ in the standard way~\cite{El-Gamal--Kim2011}. 

%Jan14: Changed the reference from NIT to Chandra's recent paper.
It is easy to see that the capacity region of this channel is the same as that of the broadcast channel with input $X$ and outputs $(Y_1,S_1)$ and $(Y_2,S_2)$. This equivalent broadcast channel, however, does not belong to any class of channels with known capacity region (see~\cite{Chandra2014} for classes of broadcast channels with known capacity). Also, very little is known about the capacity region of the broadcast channel with random state known only at the receivers. Previous work on this setting has focused mainly on the Gaussian fading BC with superposition coding~\cite{Jafarian--Vishwnath2011}, time division with power control~\cite{Liang--Goldsmith2005}, and a superposition of binary inputs motivated by a capacity achieving strategy for a layered erasure broadcast channel~\cite{tse_fadingBC}. Even when the fading BC is degraded and superposition coding is optimal, Gaussian input distribution is not in general optimal and capacity remains unknown~\cite{Abbe--Lizhong2009}. 
%Jan15. Before:
%Previous work on this setting has focused mainly on the Gaussian fading BC using superposition coding~\cite{Jafarian--Vishwnath2011}, time division with power control~\cite{Liang--Goldsmith2005}, and a superposition of binary inputs motivated by a capacity achieving strategy for a layered erasure broadcast channel~\cite{tse_fadingBC}. Even when the fading BC is degraded and superposition coding is optimal, Gaussian input distribution is not in general optimal~\cite{Abbe--Lizhong2009}.

There has been more work on the broadcast channel with causal and noncausal state information known at the transmitter. In~\cite{reza}, the capacity region of the deterministic BC when the state is known noncausally at the transmitter is established. In~\cite{lapidoth_semideterministic}, this result is extended to semideterministic BC, and it is shown that the capacity region does not enlarge when the state is also known at the receivers. %Jan15: I removed ", and several results on deterministic and semideterministic BCs with state, all of which assume that the state is known  causally or noncausally at the transmitter, are given. "

There has also been work on the setting in which the state is known at the receivers and only {\em strictly causally} at the transmitter. In~\cite{Larsson--Johansson2006,Georgiadis--Tassiulas2009}, the capacity region of the binary erasure broadcast channel with state under this setting is established. In~\cite{Ali--Tse2010} it is shown via two examples that strictly causal state information at the transmitter can enlarge the capacity region of the broadcast channel with state. In~\cite{dstate}, it is shown that the scheme in~\cite{Ali--Tse2010} is a special case of a straightforward adaptation of the feedback scheme in~\cite{Shayevitz--Wigger2013}. The Blackwell broadcast channel with state in Example~\ref{ex:blackwell} is also introduced and an achievable rate region is established when $p_1=p_2=0.5$ (and the state in known at the receivers and strictly casually at the transmitter).  %Jan14: Before: The above Blackwell broadcast channel with state

In this paper we establish the capacity region of the BC-TDCS when the state is known only at the receivers. Achievability is established using Marton coding~\cite{Marton1979}. The key observation is that the auxiliary random variables in the Marton region characterization, $U_1$ and $U_2$, are always set to $f_1,f_2,X$, or $\emptyset$. In particular if the channel from $X$ to $Y_1$ is more likely to be $f_1$ than the channel from $X$ to $Y_2$, then $(U_1, U_2)$ are set to $(X,\emptyset),(\emptyset,X)$, or $(f_1,f_2)$. The converse is established by showing that the Marton inner bound with these extreme choices of auxiliary random variables coincides with the UV outer bound~\cite{Nair--El-Gamal2006}. %Added: in the Marton region characterization

Our result is significant for several reasons:
\begin{itemize}
\item It establishes the capacity region of a new class of broadcast channels---our setting does not belong to any class of broadcast channels with previously known capacity region.
\item It establishes the capacity region of a nontrivial class of broadcast channels with state known at the receivers---a setting with very few known results.
\item It provides yet another class of broadcast channels for which Marton coding is optimal.
\item Our channel model can be used to approximate certain fading broadcast channels in high SNR (see Example~\ref{ex:finite} in Section~\ref{section2}). 
\end{itemize}

\section{Capacity Region of the BC-TDCS}\label{section2}%Jan15: Following ->The following theorem
Without loss of generality, assume $p_1 \geq p_2$. We now state the main result of this paper.
\begin{theorem}\label{capacity:Marton}
\textnormal{The capacity region of the BC-TDCS $(\Xc \times \Sc, p(s)p(y_1,y_2|x,s), \Yc_1 \times \Yc_2)$ with the state known only at the receivers is the convex hull of the set of all rate pairs $(R_1,R_2)$ such that 
%Jan14: Before: ~, $\Cr$, is the convex hull of ~
\begin{align}
\begin{split}\label{thm1}
R_1 &\leq I(U_1;Y_1|S),\\
R_2 &\leq I(U_2;Y_2|S),\\
R_1+R_2 &\leq I(U_1;Y_1|S)+I(U_2;Y_2|S)-I(U_1;U_2)
\end{split}
\end{align}
for some $p(x)$ and either $(U_1,U_2)=(f_1,f_2)$, $(U_1,U_2)=(X,\emptyset)$, or $(U_1,U_2)=(\emptyset,X)$.
}\end{theorem}
Achievability follows immediately since~\eqref{thm1} is contained in Marton's rate region. The converse is proved in Section~\ref{sec:converse}.

Now consider the following more explicit characterization of the capacity region which we will use in the examples and the converse. 
\begin{proposition} \label{capacity}
\textnormal{The capacity region of the BC-TDCS with the state known only at the receivers is the convex hull of the union of four rate regions: %Jan14: Removed , $\Cr$, 
\begin{align}
\begin{split}\label{proposition}
\Rr_1 = \{(R_1,R_2)\colon &R_1 \leq C_1, R_2 = 0\},\\
\Rr_2 = \{(R_1,R_2)\colon &R_1 =0, R_2 \leq C_2\},\\
\Rr_3 = \{(R_1,R_2)\colon &R_{1} \leq \p H(f_{1})+\bp I(f_{1};f_{2}),\\ 
					     &R_{2} \leq \bq H(f_{2}|f_{1})\\
					     &\text{for some }p(x) \in \Pc_1\},\\
\Rr_4 = \{(R_1,R_2)\colon &R_{1} \leq \p H(f_{1}|f_{2}),\\ 
					      &R_{2} \leq \q I(f_{1};f_{2})+\bq H(f_{2})\\
					      &\text{for some }p(x) \in \Pc_2\},
\end{split}
\end{align}
where $C_j = \max_{p(x)} I(X;Y_j|S)$ for $j=1,2$, and
\begin{align*}
&\Pc_1=\{\argmax_{p(x)}\, \p H(f_1) + \bp I(f_1;f_2) + \lambda \bq H(f_2|f_1)\colon\\
&\qquad\qquad\text{for some } \bp/\bq \leq \lambda \leq 1\},\\
&\Pc_2=\{\argmax_{p(x)}\, \p H(f_1|f_2) + \lambda \q I(f_1;f_2)+\lambda \bq H(f_2)\colon\\ %Jan15: Changed the order of summation
&\qquad\qquad\text{for some } 1 \leq \lambda \leq \p/\q\}.
\end{align*}}
\end{proposition}
\begin{proof} Let $\Cr$ and $\Cr_0$ denote the region defined in~\eqref{thm1} and in~\eqref{proposition}, respectively. All we need to show is that $\Cr_0=\Cr.$
First note that we can express $\Cr$ as the convex hull of the union of the four regions:
\begin{align}\begin{split}\label{primes}
\Rr'_1 = \{(R_1,R_2)\colon &R_1 \leq I(X;Y_1|S), R_2=0 \text{ for some } p(x)\},\\ %R_2=0
\Rr'_2 = \{(R_1,R_2)\colon &R_1=0, R_2 \leq I(X;Y_2|S) \text{ for some } p(x)\},\\ %R_1=0
\Rr'_3 = \{(R_1,R_2)\colon &R_1 \leq I(f_1;Y_1|S),\\
					&R_2 \leq I(f_2;Y_2|S)- I(f_1;f_2)\\
					&\text{for some } p(x) \},\\
\Rr'_4 = \{(R_1,R_2)\colon &R_1 \leq I(f_1;Y_1|S) - I(f_1;f_2),\\
					&R_2 \leq I(f_2;Y_2|S)\\
					&\text{for some } p(x)\}.
\end{split}\end{align}
Clearly $\Rr_1 = \Rr'_1$, $\Rr_2 = \Rr'_2$, $\Rr_3 \subseteq \Rr'_3$, and $\Rr_4 \subseteq \Rr'_4$. Thus, $\Cr_0\subseteq \Cr$. %Jan15
We now show that every supporting hyperplane of $\Cr$ intersects $\Cr_0$, i.e., for every  $\lambda \geq 0$, there exists a rate pair $(R_1,R_2)\in \Cr_0$ such that $R_1+\lambda R_2 = \max_{(R_1,R_2)\in \Cr}R_1+\lambda R_2$. 
%Jan15: Before: as commented.
\begin{lemma}\label{lem4}\textnormal{Every supporting hyperplane of $\Cr$ intersects $\Cr_0$, i.e., for all $\lambda \geq 0$,
\begin{align*}
&\max_{(R_1,R_2) \in \Cr}R_1+\lambda R_2 =\max_{(R_1,R_2) \in \Cr_0}R_1+\lambda R_2.
% &\qquad =\begin{cases} \max_{(R_1,R_2)\in \Rr_1}R_1+\lambda R_2 \text{ if }\lambda \leq \bp/\bq,\\
% \max_{(R_1,R_2)\in \Rr_3}R_1+\lambda R_2 \text{ if } \bp/\bq < \lambda \leq 1,\\ 
% \max_{(R_1,R_2)\in \Rr_4}R_1+\lambda R_2 \text{ if } 1 < \lambda \leq \p/\q,\\
% \max_{(R_1,R_2)\in \Rr_2}R_1+\lambda R_2 \text{ if } \lambda > \p/\q.
% \end{cases}
\end{align*}
}\end{lemma}
The proof of this lemma is in Appendix~\ref{appendix}. 

To complete the proof we use the following.
%Another option 1:
%According to Lemma~\ref{subseteq}, to show $\Cr_0=\Cr$, it is sufficient to show that every supporting hyperplane of $\Cr$ intersects with $\Cr_0$, i.e., for every  $\lambda \geq 0$, there exists a rate pair $(R_1,R_2)\in \Cr_0$ such that $R_1+\lambda R_2 = \max_{(R_1,R_2)\in \Cr}R_1+\lambda R_2$.
%Another option 2: 
%In the following we show that 
\begin{lemma}\label{subseteq}
\textnormal{~\cite{Eggleston1958} Let $\Rr \in \mathbb{R}^d$ be convex and $\Rr_1 \subseteq \Rr_2$ be two bounded convex subsets of $\Rr$, closed relative to $\Rr$. If every supporting hyperplane of $\Rr_2$ intersects $\Rr_1$, then $\Rr_1 = \Rr_2$. 
}\end{lemma}
% The proof of the proposition is completed using the following fact.
\end{proof}
\smallskip

\noindent{\bf Example \hspace{-1pt}1 (continued)} 
\textnormal{The capacity region of the Blackwell channel with state known only to the receivers is the convex hull of the union of:
\begin{align*} 
\Rr'_3 =\{(R_1,R_2)\colon\, &R_1 \leq H(\a_0) - \bp \bar{\a}_1 H(\a_0/\bar{\a}_1),\\
						&R_2 \leq \bq\bar{\a}_0 H(\a_1/\bar{\a}_0)\\
						&\text{for some } \a_0,\a_1\geq 0, \a_0+\a_1 \leq 1\},\text{ and}\\
\Rr'_4 = \{(R_1,R_2)\colon\, &R_1 \leq \p\bar{\a}_1H(\a_0/\bar{\a}_1),\\
						&R_2 \leq H(\a_1) - \q\bar{\a}_0H(\a_1/\bar{\a}_0)\\
						&\text{for some } \a_0,\a_1\geq 0, \a_0+\a_1 \leq 1\}.
\end{align*} %need to say \bar{\a}_0 = 1-\a_0 ??
%This is because  $(C_1,0)=(1,0) \in \Rr_3$ and $(0,C_2)=(0,1) \in \Rr_4$, and thus, $\Rr_1 \subseteq \Rr_3$ and $\Rr_2 \subseteq \Rr_4$.
To show this, we evaluate $\Rr'_3$ and $\Rr'_4$ in~\eqref{primes} and note that the rate pairs $(C_1,0)=(1,0) \in \Rr'_3$ and $(0,C_2)=(0,1) \in \Rr'_4$. Hence, $\Cc$ is the convex hull of the union of $\Rr'_3$ and $\Rr'_4$.  
The capacity region with state for $(p_1,p_2)=(0.5, 0.5),\, (0.7,0.3)$, and $(1,0)$ is plotted in Figure~\ref{fig2}. For $(p_1,p_2)=(0.5, 0.5)$, the two channels are statistically identical, hence the capacity region coincides with the time-division region.  For $(p_1,p_2)=(1,0)$, the channel reduces to the Blackwell channel with no state~\cite{blackwell}. For $(p_1,p_2)$ in between these two extreme cases, the capacity region is established by our theorem.} %Jan15: replaced (when) to (if~ then).
%Jan15: two channels are statistically identical --> two outputs are statistically identical.
 \begin{figure}[h] 
  \centering
\psfrag{a1}[c]{0}
\psfrag{a2}[c]{}
\psfrag{b}[c]{1}
\psfrag{c}[c]{1}
\psfrag{e}[c]{}
\psfrag{f}[c]{}
\psfrag{n}[b][b][0.9]{$\qquad(0.5,0.5)$}
\psfrag{k}[c][l][0.9]{$\qquad(0.7,0.3)$}
\psfrag{m}[r][l][0.9]{$\qquad(1,0)$}
\psfrag{g}[b]{$R_2$}
\psfrag{h}[l]{$R_1$}
    \includegraphics[scale=0.45]{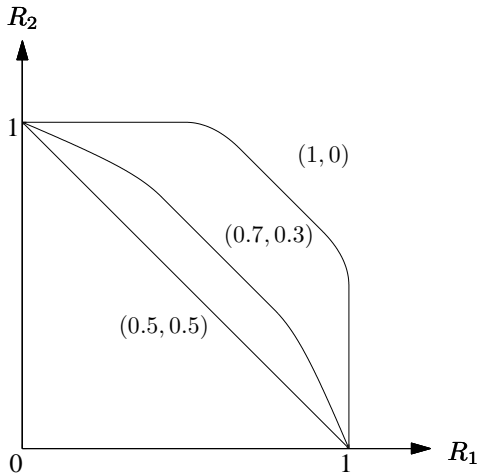}
      \caption{Capacity region of the Blackwell channel with the state. }
      \label{fig2}
 \end{figure}

Next consider the following example which is motivated by deterministic approximations of wireless channels.
%\begin{example}[Finite-field BC-TDCS~\cite{}] \label{ex:finite}: Reference?!
\begin{example}[Finite-field BC-TDCS] \label{ex:finite}
\textnormal{Consider the BC-TDCS  with the state known only at the receivers with $\mathbf{X}= \begin{bmatrix}X_1&X_2 \end{bmatrix}^T$:%Jan15: edited the form as we discussed.
\begin{align}\begin{split}\label{ex2}
Y_1 = \begin{cases}	h_{11}X_{1}+h_{12}X_2 &\text{ if } S_1=1,\\ 
				h_{21}X_{1}+h_{22}X_2 &\text{ if } S_1=2,
	\end{cases}\\
Y_2 = \begin{cases} h_{11}X_{1}+h_{12}X_2&\text{ if } S_2=1,\\
				h_{21}X_{1}+h_{22}X_2&\text{ if } S_2= 2,
	\end{cases}
\end{split}\end{align}
where the channel matrix is full-rank, $\Yc_1=\Yc_2=\Xc_1=\Xc_2=[0:K-1]$, and the arithmetic is over the finite field.} 

\textnormal{To compute the capacity region, first note that $C_1=\log K$ and  $C_2=\log K$. Thus, }
    \begin{align*}
    \Rr_1 = \{(R_1,R_2)\colon &R_1 \leq \log K, R_2 = 0 \},\\
    \Rr_2 = \{(R_1,R_2)\colon &R_1 =0, R_2 \leq \log K\}.
    \end{align*} 
\textnormal{To evaluate $\Rr_3$ and $\Rr_4$, we compute $\Pc_1$ and $\Pc_2$. Since }
    \begin{align*}
&\p H(f_1)+\bp I(f_1;f_2)+\lambda \bq H(f_2|f_1)\\
&\qquad=\p H(f_1)+\bp H(f_2)+(\lambda \bq-\bp) H(f_2|f_1)\\
&\qquad\leq (\p+\lambda\bq)\log K
\end{align*}
\textnormal{for $\bp/\bq \leq \lambda \leq 1$ with equality if $\mathbf{X} \sim \U ([0:K-1]^2)$, $\Pc_1=\left\{ \U( [0:K-1]^2)\right\}$.
Similarly, $\Pc_2=\left\{\U ( [0:K-1]^2) \right\}$.
Note that when $\mathbf{X}$ is uniform,  $H(f_{1})=H(f_{2})=H(f_{1}|f_{2})=H(f_{2}|f_{1})=\log K$. Hence,}
    \begin{align*}
    \Rr_3 = \{(R_1,R_2)\colon &R_1 \leq \p\log K, R_2 \leq \bq \log K \},\\
    \Rr_4 = \{(R_1,R_2)\colon &R_1 \leq \p\log K, R_2 \leq \bq \log K \},
    \end{align*}
\textnormal{and the capacity region is}
\begin{align*}
\Cr =\co \{(0,0),(\log K,0),(0,\log K), (\p \log K, \bq \log K)\}.
\end{align*}

\textnormal{Figure~\ref{fig3} plots the capacity region for $(p_1,p_2)=(0.5,0.5), (0.7,0.4)$, and $(1,0)$. For $(p_1,p_2)=(0.5,0.5)$, the two channels are statistically identical and the capacity region coincides with the time-division region. For $(p_1,p_2)=(1,0)$, the capacity region is $\{(R_1,R_2)\colon R_1 \le \log K,\; R_2 \le \log K\}$ because the channel matrix is full-rank. For $(p_1,p_2)$ in between these two extreme cases, the capacity region is established by our theorem.} 
% $\{R_1 \le \log K,\; R_2 \le \log K\}$ because the channel matrix is full-rank.} 

 \begin{figure}[h!]
\raggedright
 \psfrag{a1}[c]{0}
\psfrag{a2}[c]{}
\psfrag{b}[c]{1}
\psfrag{c}[r]{1}
\psfrag{e}[c]{0.7}
\psfrag{f}[r]{0.6}
\psfrag{n}[c][][0.9]{$\qquad(0.5,0.5)$}
\psfrag{k}[c][][0.9]{$\qquad(0.7,0.4)$}
\psfrag{m}[c][][0.9]{$\qquad(1,0)$}
\psfrag{g}[b]{$R_2/\log K$}
\psfrag{h}[l]{$R_1/\log K$}
 \includegraphics[scale=0.45]{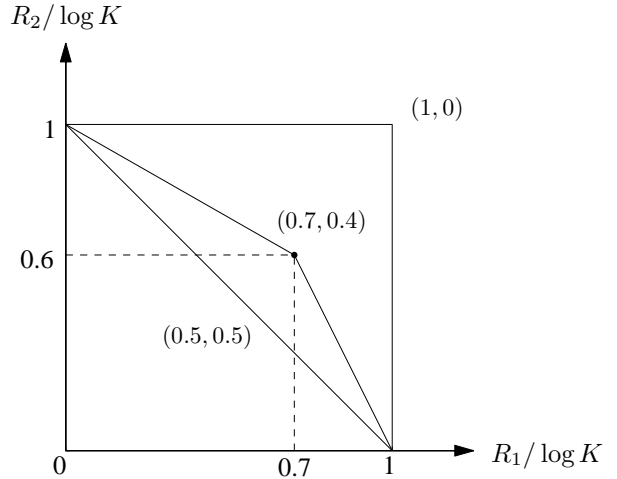}
 \caption{Capacity region of the Finite Field BC-TDCS.}\label{fig3}
 \end{figure}

\noindent Connection to wireless channels: %Jan15: Italic?
%\noindent{\em Connection to wireless channels}: 
\textnormal{Consider the following fading broadcast channel
\begin{align}\label{fading}
Y_j=\mathbf{H}_j^{\dagger}\mathbf{X}+Z_j\text{ for } j=1,2,
\end{align}
where $\dagger$ denotes the conjugate-transpose, $\mathbf{X}=\begin{bmatrix}X_1&X_2 \end{bmatrix}^T\in \mathbb{C}^{2\times 1}, \mathbb{E} [\mathbf{X}^{\dagger}\mathbf{X}] \leq P$, $Z_j \sim \mathcal{C}\mathcal{N}(0,1)$ and the noise sequences $Z_{ji}$, $j=1,2$ and $i\in [1:n]$, are i.i.d. In addition, for $j=1,2$, 
\begin{align*}
\mathbf{H}^{\dagger}_j =\begin{cases}
					[ h_{11} \quad h_{12}]& \text{if }S_j=1 \text{ w.p. } p_j,\\
					[h_{21}  \quad h_{22}] & \text{if }S_j=2 \text{ w.p. } \bar{p}_j,
					\end{cases}
\end{align*}%
where the channel matrix is in $\mathbb{C}^{2\times 2}$ and is full rank.}

\textnormal{We now show that the degrees of freedom (DoF) of this fading Gaussian broadcast channel, obtained by dividing the maximum sum-rate by $\log P$ and taking the limit, is $\p+\bq$.}

\textnormal{Since the variance of the noise $Z_j$ is bounded, the DoF of channel in~\eqref{fading} is equal to that of the BC-TDCS with $Y_j=\mathbf{H}_j^{\dagger}\mathbf{X}$ for $j=1,2$~\cite{Ali--Tse2010}.}
%Jan15: Before: $f_1(\mathbf{X}),f_2(\mathbf{X})$ defined as in~\eqref{ex2}.}
\textnormal{We show that the DoF is achieved when $U_1=f_1$ and $U_2=f_2$ are independent and Gaussian with variances $\alpha P$ and $\beta P$ for some $\alpha,\beta>0$ such that
\begin{align*}
\begin{bmatrix}X_1\\X_2\end{bmatrix}=\begin{bmatrix}
    h_{11} & h_{12}\\
    h_{21} & h_{22}
    \end{bmatrix} ^{-1}\begin{bmatrix}U_1\\U_2\end{bmatrix}
\end{align*}
satisfy the power constraint. First note that for $(R_1,R_2) \in \Cr$,
\begin{align}
&\max \lim _{P \to \infty}\frac{R_1+R_2}{\log P}\nonumber\\
&=\max_{p(\mathbf{X})} \lim _{P \to \infty}\frac{\p H(f_1)+ \bq H(f_2)+(\bp-\bq)I(f_1;f_2)}{\log P}.\label{3terms}
\end{align}}

\textnormal{Now we show that each term in~\eqref{3terms} is maximized with the chosen input. First, $\lim _{P \to \infty}\p H(f_1)/\log P = \lim_{P \to \infty}\p \log(\alpha P)/\log P = p_1$. Now we show that $\p = \max \lim _{P \to \infty}\p H(f_1)/\log P$. Since  $\var(f_1)=\var(h_{11}X_1+h_{12}X_2) = |h_{11}|^2 \gamma P+|h_{12}|^2 \bar{\gamma} P+(h_{11}^*h_{12}+h_{12}^*h_{11})\rho \sqrt{\gamma \bar{\gamma}}P$ for some $0\leq \gamma,\rho \leq 1$ due to the power constraint, $H(f_1) \leq \log(|h_{11}|^2 \gamma+|h_{12}|^2 \bar{\gamma}+(h_{11}^*h_{12}+h_{12}^*h_{11})\rho\sqrt{\gamma \bar{\gamma}})+\log P$. Hence, $\lim _{P \to \infty}\p H(f_1)/\log P \leq \p$.
Similarly, $\lim _{P \to \infty}\bq H(f_2)/\log P$ is maximized and is equal to $\bq$, and $\lim _{P \to \infty}(\bp-\bq)I(f_1;f_2)/\log P$ is maximized and is equal to $0$. Thus, the following holds: 
\begin{align*}
&\max_{p(\mathbf{X})} \lim _{P \to \infty}\frac{\p H(f_1)+ \bq H(f_2)+(\bp-\bq)I(f_1;f_2)}{\log P} \\
&\qquad= \p+\bq,
\end{align*}
and the DoF of the fading Gaussian BC in~\eqref{fading} is $\p+\bq$.}
\end{example}

\section{Proof of the Converse}\label{sec:converse}
The UV bound for the broadcast channel with state known at the receivers states that if a rate pair $(R_1,R_2)$ is achievable, then it must lie in the intersection of the regions
\begin{align*}
\bar{\Rr}_{1} = \{(R_1,R_2)\colon &R_1 \leq I(U_1;Y_1|S),\\
&R_2 \leq I(X;Y_2|S),\\
&R_1+R_2 \leq I(U_1;Y_1|S)+I(X;Y_2|U_1,S) \\
&\text{for some } p(u_1,x)\},\\ %\\
\bar{\Rr}_{2} = \{(R_1,R_2)\colon &R_1 \leq I(X;Y_1|S),\\
&R_2 \leq I(U_2;Y_2|S),\\
&R_1+R_2 \leq I(U_2;Y_2|S)+I(X;Y_1|U_2,S)\\
&\text{for some } p(u_2,x)\}.
\end{align*}
Denote this outer bound by $\bar \Rr$.

% Jan 13.     
%To establish the converse we show that for every  $\lambda \geq 0$, there exists a rate pair $(R_1,R_2)\in \Cr$ such that $R_1+\lambda R_2 = \max_{(R_1,R_2)\in \bar{\Rr}}R_1+\lambda R_2$,  that is, every supporting hyperplane of $\bar{\Rr}$ intersects with $\Cr$. 
To establish the converse we show that every supporting hyperplane of $\bar{\Rr}$ intersects $\Cr$.
\begin{lemma}\label{lem:converse}\textnormal{For all $\lambda \geq 0$,
\begin{align}
\max_{(R_1,R_2)\in\bar{\Rr}} R_1+\lambda R_2 =\max_{(R_1,R_2)\in\Cr} R_1+\lambda R_2.
\end{align}
}\end{lemma}
\begin{proof}
To prove the lemma for $0\leq \lambda \leq 1$, consider maximizing $R_1+\lambda R_2$ over $(R_1,R_2) \in \bar{\Rr}_1$. 

For any $p(u_1,x)$, $R_1+\lambda R_2$ such that $(R_1,R_2) \in \bar \Rr_1$ is maximized when $R_1=I(U_1;Y_1|S)$ and $R_2=I(X;Y_2|U_1,S)=H(Y_2|U_1,S)$. Thus,
\begin{align}
&\max_{(R_1,R_2) \in \bar{\Rr}_1}R_1+\lambda R_2=\max_{p(u_1,x)} I(U_1;Y_1|S)+\lambda H(Y_2|U_1,S)\nonumber\\
&=\max_{p(x)}\big\{H(Y_1|S)+\max_{p(u_1|x)}\{ \lambda H(Y_2|U_1,S)-H(Y_1|U_1,S)\} \big\}\nonumber\\%\label{g:eq2}\\
&=\max_{p(x)}\big\{\p H(f_1)+ \bp H(f_2)  +\max_{p(u_1|x)}\{ (\lambda \bq-\bp)H(f_2|U_1) \nonumber\\
&\qquad \qquad  + (\lambda \q-\p)H(f_1|U_1)\} \big\}.\label{g:eq4}
\end{align}
For a fixed $p(x)$ only the last two terms in~\eqref{g:eq4} depend on $p(u_1|x)$.
We now consider different ranges of $0\leq \lambda \le 1$. 
\begin{itemize}
\item If $0\leq \lambda \leq \bp/\bq$, then for any fixed $p(x)$, 
\begin{align*}
(\lambda\bq-\bp)H(f_2|U_1) + (\lambda \q-\p)H(f_1|U_1) \leq 0
\end{align*}
with equality if $U_1=X$. Thus,~\eqref{g:eq4} can be rewritten as
\begin{align*}
\max_{(R_1,R_2) \in \bar{\Rr}_1}R_1+\lambda R_2 &=\max_{p(x)}\p H(f_1)+\bp H(f_2)\\
										 &=\max_{(R_1,R_2)\in \Rr_1}R_1+\lambda R_2
\end{align*}
\item If $\bp/\bq<\lambda \leq 1$, then for any fixed $p(x)$,
\begin{align*}
&(\lambda\bq-\bp)H(f_2|U_1) + (\lambda \q-\p)H(f_1|U_1)\nonumber\\
% &\qquad\qquad = (\lambda\bq-\bp)\{H(f_1|U_1)+H(f_2|f_1,U_1)-H(f_1|f_2,U_1)\}-(\p-\lambda \q)H(f_1|U_1)\nonumber\\
&=(\lambda-1)H(f_1|U_1)\\
&\qquad\qquad+(\lambda\bq-\bp) (H(f_2|f_1,U_1)-H(f_1|f_2,U_1))\nonumber\\
&\leq (\lambda\bq-\bp)H(f_2|f_1)
\end{align*}
with equality if $U_1=f_1$. Thus,~\eqref{g:eq4} can be rewritten as
\begin{align*}
&\max_{(R_1,R_2) \in \bar{\Rr}_1}R_1+\lambda R_2 \\
%Jan 9: deleted the following line
% &\qquad=\max_{p(x)}\p H(f_1)+\bp H(f_2)+(\lambda \bq-\bp)H(f_2|f_1)\\
&\qquad=\max_{p(x)}\p H(f_1)+\bp I(f_1;f_2)+\lambda\bq H(f_2|f_1)\\
&\qquad=\max_{(R_1,R_2)\in \Rr_3} R_1+\lambda R_2.
\end{align*}
\end{itemize}
Thus, $\max_{(R_1,R_2)\in \Cr}R_1+\lambda R_2 \geq \max_{(R_1,R_2)\in \bar{\Rr}}R_1+\lambda R_2 $ for $0\leq \lambda\leq 1$. Equality in the lemma holds because $\Cr \subseteq \bar{\Rr}$.
The proof for $\lambda >1$ follows similarly (see Appendix~\ref{appendix2}).
\end{proof}

The proof of the converse is completed using Lemma~\ref{subseteq}.

\section{Conclusion}        
We established the capacity region of the BC-TDCS channel when the state is known only at the receivers. This channel does not belong to any class of broadcast channels for which the capacity was previously known. %This channel model can be also used to obtain the optimal degree of freedom for certain fading broadcast channels.
% If $Y_1=f_1(X), Y_2=f_2(X),$ then Marton coding with $(U_1,U_2)=(f_1(X),f_2(X))$ is optimal. 
There are several open problems that would be interesting to explore further, including: What is the capacity region of the BC-TDCS with common message when the state is known only at the receivers?
%\item What is the capacity region when the channel state is not known to the transmitter or receivers?
What is the capacity region when each channel component is selected from a set of more than two deterministic channel states?

\section{Acknowledgments} 
The authors thank Chandra Nair and Young-Han Kim for comments that have improved the readability of this paper.

\bibliographystyle{IEEEtran}
\bibliography{bcstate}

\begin{thebibliography}{10}
\providecommand{\url}[1]{#1}
\csname url@samestyle\endcsname
\providecommand{\newblock}{\relax}
\providecommand{\bibinfo}[2]{#2}
\providecommand{\BIBentrySTDinterwordspacing}{\spaceskip=0pt\relax}
\providecommand{\BIBentryALTinterwordstretchfactor}{4}
\providecommand{\BIBentryALTinterwordspacing}{\spaceskip=\fontdimen2\font plus
\BIBentryALTinterwordstretchfactor\fontdimen3\font minus
  \fontdimen4\font\relax}
\providecommand{\BIBforeignlanguage}[2]{{%
\expandafter\ifx\csname l@#1\endcsname\relax
\typeout{** WARNING: IEEEtran.bst: No hyphenation pattern has been}%
\typeout{** loaded for the language `#1'. Using the pattern for}%
\typeout{** the default language instead.}%
\else
\language=\csname l@#1\endcsname
\fi
#2}}
\providecommand{\BIBdecl}{\relax}
\BIBdecl

\bibitem{dstate}
H.~Kim, Y.-K. Chia, and A.~El~Gamal, ``A note on broadcast channels with stale
  state information at the transmitter,'' \emph{CoRR}, vol. abs/1309.7437,
  2013.

\bibitem{El-Gamal--Kim2011}
A.~El~Gamal and Y.~H. Kim, \emph{Network Information Theory}, 1st~ed.\hskip 1em
  plus 0.5em minus 0.4em\relax Cambridge University Press, 2011.

\bibitem{Chandra2014}
Y.~Geng, A.~Gohari, C.~Nair, and Y.~Yu, ``On marton's inner bound and its
  optimality for classes of product broadcast channels,'' \emph{Information
  Theory, IEEE Transactions on}, vol.~60, no.~1, pp. 22--41, 2014.

\bibitem{Jafarian--Vishwnath2011}
A.~Jafarian and S.~Vishwanath, ``The two-user gaussian fading broadcast
  channel,'' in \emph{Information Theory Proceedings (ISIT), 2011 IEEE
  International Symposium on}, 2011, pp. 2964--2968.

\bibitem{Liang--Goldsmith2005}
Y.~Liang and A.~Goldsmith, ``Rate regions and optimal power allocation for td
  fading broadcast channels without csit,'' in \emph{Allerton, Monticello IL},
  2005.

\bibitem{tse_fadingBC}
D.~Tse and R.~Yates, ``Fading broadcast channels with state information at the
  receivers,'' \emph{Information Theory, IEEE Transactions on}, vol.~58, no.~6,
  pp. 3453--3471, 2012.

\bibitem{Abbe--Lizhong2009}
E.~Abbe and L.~Zheng, ``Coding along hermite polynomials for gaussian noise
  channels,'' in \emph{Information Theory, 2009. ISIT 2009. IEEE International
  Symposium on}, 2009, pp. 1644--1648.

\bibitem{reza}
R.~Khosravi-Farsani and F.~Marvasti, ``Capacity bounds for multiuser channels
  with non-causal channel state information at the transmitters,'' in
  \emph{Information Theory Workshop (ITW), 2011 IEEE}, 2011, pp. 195--199.

\bibitem{lapidoth_semideterministic}
A.~Lapidoth and L.~Wang, ``The state-dependent semideterministic broadcast
  channel,'' \emph{Information Theory, IEEE Transactions on}, vol.~59, no.~4,
  pp. 2242--2251, 2013.

\bibitem{Larsson--Johansson2006}
P.~Larsson and N.~Johansson, ``Multi-user arq,'' in \emph{Vehicular Technology
  Conference, 2006. VTC 2006-Spring. IEEE 63rd}, vol.~4, 2006, pp. 2052--2057.

\bibitem{Georgiadis--Tassiulas2009}
L.~Georgiadis and L.~Tassiulas, ``Broadcast erasure channel with feedback -
  capacity and algorithms,'' in \emph{Network Coding, Theory, and Applications,
  2009. NetCod '09. Workshop on}, 2009, pp. 54--61.

\bibitem{Ali--Tse2010}
M.~Maddah-Ali and D.~Tse, ``Completely stale transmitter channel state
  information is still very useful,'' in \emph{Information Theory, IEEE
  Transactions on}, vol.~58, no.~7, 2012, pp. 4418--4431.

\bibitem{Shayevitz--Wigger2013}
O.~Shayevitz and M.~Wigger, ``On the capacity of the discrete memoryless
  broadcast channel with feedback,'' \emph{{IEEE} Trans. Inf. Theory}, vol.~59,
  no.~3, pp. 1329--1345, 2013.

\bibitem{Marton1979}
K.~Marton, ``A coding theorem for the discrete memoryless broadcast channel,''
  \emph{Information Theory, IEEE Transactions on}, vol.~25, no.~3, pp.
  306--311, 1979.

\bibitem{Nair--El-Gamal2006}
C.~Nair and A.~El~Gamal, ``An outer bound to the capacity region of the
  broadcast channel,'' in \emph{Information Theory, 2006 IEEE International
  Symposium on}, 2006, pp. 2205--2209.

\bibitem{Eggleston1958}
H.~G. Egglestone, \emph{Convexity}.\hskip 1em plus 0.5em minus 0.4em\relax
  Cambridge University Press, Cambridge, 1958.

\bibitem{blackwell}
D.~Blackwell, L.~Breiman, and A.~J. Thomasian, ``Proof of shannon's
  transmission theorem for finite-state indecomposable channels,'' \emph{The
  Annals of Mathematical Statistics}, vol.~29, no.~4, pp. pp. 1209--1220, Dec.
  1958.

\end{thebibliography}

\appendices

\section{Proof of Lemma~\ref{lem4}}\label{appendix}
 We prove the lemma for $0\leq \lambda \leq 1$. The proof for $\lambda >1$ follows similarly. First we show that for $0\leq \lambda \leq 1$,
\begin{align*}%\label{eq1}
&\max_{(R_1,R_2)\in \Cr}R_1+\lambda R_2 = \max_{i=1,2,3,4}\big\{\max_{(R_1,R_2)\in \Rr'_i}R_1+\lambda R_2\big\}\\
&\qquad\stackrel{(a)}{=}\max_{i=1,3}\bigg\{\max_{(R_1,R_2)\in\Rr'_i}R_1+\lambda R_2\bigg\}\\
&\qquad\stackrel{(b)}{=}\begin{cases} \max_{(R_1,R_2)\in \Rr_1}R_1+\lambda R_2 \text{ if }0 \leq \lambda \leq \bp/\bq,\\
\max_{(R_1,R_2)\in \Rr_3}R_1+\lambda R_2 \text{ if } \bp/\bq < \lambda \leq 1.\\ 
\end{cases}
\end{align*}

The equality in $(a)$ holds because \begin{align*}
&\max_{(R_1,R_2)\in \Rr'_3}R_1+\lambda R_2\\
&\quad =\max_{p(x)}I(f_1;Y_1|S)+\lambda I(f_2;Y_2|S)- \lambda I(f_1;f_2) \\
&\quad \geq \max_{p(x)}I(f_1;Y_1|S)- I(f_1;f_2)+\lambda I(f_2;Y_2|S)\\
&\quad =\max_{(R_1,R_2)\in \Rr'_4}R_1+\lambda R_2, \text{ and}\\
&\max_{(R_1,R_2)\in \Rr'_4}R_1+\lambda R_2\\
&\quad =\max_{p(x)} \lambda\q H(f_1)+\lambda\bq H(f_2)+(\p-\lambda\q)H(f_1|f_2)\\
&\quad \geq \max_{p(x)} \lambda\q H(f_1)+\lambda\bq H(f_2)=\max_{(R_1,R_2)\in \Rr'_2}R_1+\lambda R_2.
\end{align*}
To derive the equality in $(b)$ note that
\begin{align*}
&\max_{(R_1,R_2)\in \Rr'_1} R_1+\lambda R_2= \max_{p(x)}\{\p H(f_1)+\bp H(f_2)\},\\
&\max_{(R_1,R_2)\in \Rr'_3} R_1+\lambda R_2\\
&\qquad=\max_{p(x)}\{\p H(f_1)+\bp H(f_2) +(\lambda\bq-\bp)H(f_2|f_1)\}.
\end{align*}
For $0\leq \lambda\leq \bp/\bq$, $\max_{(R_1,R_2)\in\Rr'_1}R_1+\lambda R_2 \geq \max_{(R_1,R_2)\in \Rr'_3} R_1+\lambda R_2$, and for $\bp/\bq < \lambda \leq 1$, $\max_{(R_1,R_2)\in\Rr'_3}R_1+\lambda R_2\geq \max_{(R_1,R_2)\in \Rr'_1} R_1+\lambda R_2$. Finally, $(b)$ holds since $\max_{(R_1,R_2)\in\Rr'_3}R_1+\lambda R_2 = \max_{(R_1,R_2)\in\Rr_3}R_1+\lambda R_2$ for $\bp/\bq < \lambda \leq 1$ and $\Rr'_1=\Rr_1$.

Thus, for $0\leq \lambda \leq 1$,  $\max_{(R_1,R_2)\in \Cr_0}R_1+\lambda R_2 \geq \max_{(R_1,R_2)\in \Cr}R_1+\lambda R_2$. Finally, equality holds because $\Cr_0\subseteq \Cr$.

\section{Proof of Lemma~\ref{lem:converse} for $\lambda>1$}\label{appendix2}

For $\lambda > 1$, we consider the equivalent maximization problem: $\max_{(R_1,R_2) \in \bar{\Rr}_2}\lambda^{-1}R_1+ R_2$ for $\lambda^{-1} <1$.

For any $p(u_2,x)$, among the $(R_1,R_2) \in \bar{\Rr}_2$, $ \lambda^{-1}R_1+R_2$ is maximized when $R_2=I(U_2;Y_2|S)$ and $R_1=H(Y_1|U_2,S)$. Thus,
\begin{align}
&\max_{(R_1,R_2) \in \bar{\Rr}_2}\lambda^{-1}R_1+R_2\nonumber\\
&\qquad =\max_{p(u_2,x)} \lambda^{-1} H(Y_1|U_2,S)+I(U_2;Y_2|S)\nonumber\\
&\qquad =\max_{p(x)}\big\{H(Y_2|S)+\max_{p(u_2|x)}\{\lambda^{-1} H(Y_1|U_2,S) \nonumber \\
&\qquad \qquad -H(Y_2|U_2,S)\} \big\}\nonumber\\%\label{gg:eq2}\\
&\qquad =\max_{p(x)}\big\{\q H(f_1)+ \bq H(f_2) \nonumber\\
&\qquad  \qquad \qquad  +\max_{p(u_2|x)}\{(\lambda^{-1} \bp-\bq)H(f_2|U_2)\nonumber\\
&\qquad \qquad \qquad \qquad + (\lambda^{-1}\p-\q)H(f_1|U_2)\} \big\}.\label{gg:eq4}
\end{align}
For a fixed $p(x)$,  only the last two terms in~\eqref{gg:eq4} depend on $p(u_2|x)$. We now consider different ranges of $\lambda > 1$.

\begin{itemize}
\item If $\lambda > \p/\q$, then for any fixed $p(x)$, 
\begin{align*}
(\lambda^{-1}\bp-\bq)H(f_2|U_2) + (\lambda^{-1} \p-\q)H(f_1|U_2) \leq 0
\end{align*}
with equality if $U_2=X$. Then,~\eqref{gg:eq4} can be expressed as
\begin{align*}
\max_{(R_1,R_2) \in \bar{\Rr}_2}R_1+\lambda R_2 &=\max_{p(x)}\lambda \q H(f_1)+\lambda \bq H(f_2)\\
										&=\max_{(R_1,R_2)\in\Rr_2}R_1+\lambda R_2.
\end{align*}
\item If $1 < \lambda \leq \p/\q$, then for any fixed $p(x)$,
\begin{align*}
&(\lambda^{-1}\bp-\bq)H(f_2|U_2) + (\lambda^{-1}\p-\q)H(f_1|U_2)\nonumber\\
&=(\lambda^{-1}-1)H(f_2|U_2)\nonumber\\
&\qquad+(\lambda^{-1}\p-\q)\{H(f_1|f_2,U_2)-H(f_2|f_1,U_2)\}\nonumber\\
&\leq (\lambda^{-1}\p-\q)H(f_1|f_2)
\end{align*}
with equality if $U_2=f_2$. Then,~\eqref{gg:eq4} can be expressed as
\begin{align*}
&\max_{(R_1,R_2) \in \bar{\Rr}_2}R_1+ \lambda R_2\\
%Jan 9: deleted the following line
% &\quad = \max_{p(x)}\lambda\q H(f_1)+\lambda\bq H(f_2)+(\p-\lambda\q)H(f_1|f_2)\\
&\qquad = \max_{p(x)}\p H(f_1|f_2)+\lambda\q I(f_1;f_2) +\lambda\bq H(f_2)\\
&\qquad=\max_{(R_1,R_2)\in\Rr_4}R_1+\lambda R_2.
\end{align*}
\end{itemize}

Thus, $\max_{(R_1,R_2)\in \Cr}R_1+\lambda R_2 \geq \max_{(R_1,R_2)\in \bar{\Rr}}R_1+\lambda R_2$ for $\lambda > 1$. Finally, equality holds because $\Cr \subseteq \bar{\Rr}$.

\end{document}